\documentclass{birkjour}

\usepackage{amssymb,amsfonts,amsmath,amsthm}
\usepackage{dsfont}
\newtheorem{thm}{Theorem}
\newtheorem{prop}[thm]{Proposition}
\newtheorem{lem}[thm]{Lemma}
\newtheorem{assum}{Assumption}

\DeclareMathOperator{\sgn}{sgn}

\DeclareMathOperator{\artanh}{artanh}
\newcommand{\wlim}{\mathrm{w}\!-\!\lim}
\newcommand{\<}{\langle}
\renewcommand{\>}{\rangle}
\newcommand{\fm}{\!-\!}

\newcommand{\Sc}{\mathcal{S}}

\newcommand{\F}{\mathcal{F}}

\newcommand{\B}{\mathcal{B}}
\newcommand{\C}{\mathcal{C}}
\newcommand{\Hc}{\mathcal{H}}
\newcommand{\K}{\mathcal{K}}

\newcommand{\N}{\mathcal{N}}
\newcommand{\M}{\mathcal{M}}
\newcommand{\V}{\mathcal{V}}
\newcommand{\as}{\mathrm{as}}
\newcommand{\reg}{\mathrm{reg}}
\newcommand{\sing}{\mathrm{sing}}

\newcommand{\al}{\alpha}
\newcommand{\be}{\beta}
\newcommand{\ka}{\kappa}
\newcommand{\ep}{\epsilon}
\newcommand{\ga}{\gamma}
\newcommand{\la}{\lambda}
\newcommand{\vp}{\varphi}

\newcommand{\mN}{\mathbb{N}}

\newcommand{\mZ}{\mathbb{Z}}
\newcommand{\dsp}{\displaystyle}
\newcommand{\txt}{\textstyle}
\newcommand{\con}{\mathrm{const}}
\newcommand{\ti}{\widetilde}
\newcommand{\w}{\omega}
\newcommand{\W}{\Omega}
\newcommand{\ov}{\overline}

\newcommand{\p}{\partial}
\newcommand{\dV}{\dot{V}}
\DeclareMathOperator{\id}{id} 

\begin{document}

\title[Spacelike localization of long-range fields]
{Spacelike localization of long-range fields\\ in
a model of asymptotic electrodynamics}
\author[A. Herdegen]{Andrzej Herdegen}
\address{Institute of Physics\\ Jagiellonian University\\
Reymonta 4\\ 30-059 Cracow\\ Poland}
\email{herdegen@th.if.uj.edu.pl}

\author[K. Rejzner]{Katarzyna Rejzner}
\address{II. Institute for Theoretical Physics\\ Hamburg
University\\ Luruper Chaussee 149\\ 22761 Hamburg\\ Germany}
\email{katarzyna.rejzner@desy.de}

\begin{abstract}
A previously proposed algebra of asymptotic fields in quantum
electrodynamics is formulated as a net of algebras localized in
regions which in general have unbounded spacelike extension.
Electromagnetic fields may be localized in `symmetrical
spacelike cones', but there are strong indications this is
not possible in the present model for charged fields, which
have tails extending in all space directions. Nevertheless, products of appropriately `dressed' fermion fields (with compensating charges) yield bi-localized observables.
\end{abstract}

\maketitle

\renewcommand{\theequation}{\thesection.\arabic{equation}}

\section{Introduction}

In this paper we continue the investigation of the infrared
structure of quantum electrodynamics based on an algebraic
model proposed earlier by one of us (see Ref.\ \cite{her98} and
papers cited therein; see also \cite{her05}). This model is
supposed to describe asymptotic fields in the quantum
Maxwell-Dirac system, including the Gauss' law constraint (as
opposed to the crossed product of free fields).

In a recent paper \cite{her08} this model was investigated in
respect of the localization properties of fields. It was shown
that one needs an extension of the localization regions:
infrared/charge structure is encoded in unbounded regions. It was
argued that from the point of view of scattering theory, the
natural choice for extended localization regions consists of
`fattened lightcones', unions of intersecting: a future- and a
past-lightcone. The test-functions of electromagnetic fields have
well-defined asymptotes encoding the information on the long
distance structure.

In the present article we show that the algebra can be
localized in any `time-slice' which is fattening under constant
inclination towards infinity. In addition, the localization of
electromagnetic field may be restricted to `fattened
symmetrical spacelike cones': the unions of a spacelike cone
and its reflection with respect to a~point in its inside.
Similar restriction seems to be ruled out, even asymptotically,
for charged fields. This seems to contradict general wisdom on the
expected behavior of fields in full electrodynamics, see e.g.\
the assumptions on which Buchholz \cite{bu82} bases his selection
criterion of representations in quantum electrodynamics. Whether this points to some incompleteness of the model is an open question; see the discussion at the beginning of Section \ref{locdir} below and in Conclusions. On the other hand, we show that in the present model, in agreement with the general expectation, one can superpose two appropriately ``dressed'' Dirac fields carrying opposite charges to obtain a local observable.

This article should be regarded as a continuation of Refs.
\cite{her98} and \cite{her08}, and we refer the reader to these
references for more detail and a wider background.
However, we briefly summarize notation and the formulation of the model in the next two sections. We obtain spacelike localization of fields in Sections 4 and 5, and discuss the results in concluding Section 6.

\setcounter{equation}{0}

\section{Geometrical preliminaries}

The geometry of the spacetime is given by the affine Minkowski
space $\M$. If a~ref\-erence point $O$ is chosen, then each
point $P$ in $\M$ is represented by a~vector $x$ in the
associated Minkowski vector space $M$ according to \mbox{$P=O+x$}. We mostly keep $O$ fixed and use this representation. The Minkowski product is denoted by a dot, $x\cdot y$, and we write $x^2=x\cdot x$. If a~Minkowski basis $(e_0,\ldots,e_3)$ in $M$ is chosen, then we
denote \mbox{$x=x^ae_a$}. We also then use the standard
multi-index notation $x^\al=(x^0)^{\al_0}\!\ldots
(x^3)^{\al_3}$, \mbox{$|\al|=\al_0+\ldots+\al_3$},
$D^\be=\p_0^{\be_0}\!\ldots\,\p_3^{\be_3}$, where $\p_a=\p/\p
x^a$. We associate with the chosen Minkowski basis a Euclidean
metric with unit matrix in that basis, and denote by $|x|$ the
norm of $x$ in that metric. We briefly recall the definitions
of test functions spaces used in \cite{her08}. Let $\phi(x)$ be
a smooth tensor or spinor field (with vector representation of
points) and define for $\ka\geq0$, $l=0,1,\dots$ the seminorms
\begin{equation}
 \|\vp\|_{\ka,l}=\sup(1+|x|)^\ka|D^\be\vp_j(x)|\,,
\end{equation}
where supremum is taken over $x\in M$, all $\be$ such that
$|\be|=l$ and $j$ running over the components of the field.
Then $\Sc_\ka$ is the space of all smooth fields of a given
geometrical type for which all seminorms $\|.\|_{\ka+l,l}$ with
fixed $\ka$ are finite. Denote moreover the operators on smooth
functions $H=x\cdot\partial$ and $H_\ka=H+\ka\id$. Then the
space $\Sc^\ka_{\ka+\ep}$ consists of all fields which under
the action of $H_\ka$ fall into $\Sc_{\ka+\ep}$. Each field
$\vp\in\Sc_{\ka+\ep}^\ka$ has an asymptote
\begin{equation}
 \vp_\as(x)=\lim_{R\to\infty}R^\ka\vp(Rx)\,.
\end{equation}
The inversion formulas are
\begin{equation}
 \vp(x)=\int_0^1u^{\ka-1}[H_\ka\vp](ux)\,du\,,\quad
 \vp_\as(x)=\int_0^\infty u^{\ka-1}[H_\ka\vp](ux)\,du\,.
\end{equation}
The subspaces $\Sc_\ka(\W)$, $\Sc_{\ka+\ep}^\ka(\W)$ consist of
functions supported in $\W$. All spaces, as well as asymptotes,
are independent of the choice of an origin and a basis.

Next, we recall some notation for Lorentz invariant
hypersurfaces. We denote by $l$ vectors on the future
lightcone, and we also introduce
\[
 L_{ab}=l_a(\p/\p l^b)-l_b(\p/\p l^a)\,,
\]
which is an operator conveniently expressing differentiation on
the lightcone. We denote by $d^2l$ the invariant measure on the
set of null directions, which is applicable to functions $f(l)$
homogeneous of degree $-2$: the integral
\begin{equation}\label{d2l}
 \int f(l)\,d^2l=\int f(e_0+\vec{\vspace{1pt}l})\,d\W(\vec{\vspace{1pt}l})\,,
\end{equation}
where $d\W(\vec{\vspace{1pt}l})$ is the solid angle measure in the
direction of the unit 3-vector $\vec{\vspace{1pt}l}$, is
independent of the choice of Minkowski basis, and satisfies
\begin{equation}\label{parts}
 \int L_{ab}f(l)\,d^2l=0\,.
\end{equation}
We denote by $H_+$ the hyperboloid $v^2=1$, $v^0>0$. The
differentiation within the hyperboloid is conveniently expressed
by the action of the operator $\delta_a$, and integration with the
use of invariant measure $d\mu$, defined respectively by
\[
 \delta_b=v^a\big[v_a(\p/\p v^b)-v_b(\p/\p v^a)\big]\,,\quad
 d\mu(v)=2\theta(v^0)\delta(v^2-1)\,d^4v\,.
\]
We note that for a differentiable function $f(v)$ vanishing for
$v^0\to\infty$ as $o((v^0)^{-3})$, we have
\begin{equation}\label{vsto}
 \int(\delta-3v)f(v)\,d\mu(v)=0\,.
\end{equation}
For $x$ inside the future lightcone, one can write $x=\la v$, $\la>0$,  and then differentiation and integration over the inside of the future lightcone may be written as
\begin{gather}
 \p/\p x^a=v_a\p_\la+(1/\la)\delta_a\,,\label{vd}\\
 \int F(x)\,d^4x=\int F(\la v)\la^3\,d\la\,d\mu(v)\,.\label{vi}
\end{gather}
Similarly, for the hyperboloid $H_-$ formed by $z^2=-1$, the
differentiation operator and the integration measure are defined,
respectively, by
\[
 \delta_b=-z^a\big[z_a(\p/\p z^b)-z_b(\p/\p z^a)\big]\,,\quad
 d\nu(z)=2\delta(z^2+1)\,d^4z\,.
\]
For $f(z)$ vanishing for $|\vec{\vspace{1pt}z}|\to\infty$ as
$o(|\vec{\vspace{1pt}z}|^{-3})$, there is
\begin{equation}\label{zsto}
 \int(\delta+3z)f(z)\,d\nu(z)=0\,,
\end{equation}
and for $x=\la z$ ($\la>0$) running over the outside of the lightcone, the
analogues of \eqref{vd} and \eqref{vi} are
\begin{gather}
 \p/\p x^a=-z_a\p_\la+(1/\la)\delta_a\,,\label{zd}\\
 \int F(x)\,d^4x=\int F(\la z)\la^3\,d\la\,d\nu(z)\,.\label{zi}
\end{gather}

Finally, we define some spacetime sets used in the article. For
$\ga>0$ and $\delta\in(0,1)$ we shall denote by
$R_{\ga,\delta}$ the region
$|x^0|\leq\ga+\delta|\vec{\vspace{1pt}x}|$ and by $R_\delta$
the region $|x^0|\leq\delta|\vec{\vspace{1pt}x}|$. We note that
\begin{equation}\label{loreu}
 -x^2\geq\frac{1-\delta^2}{1+\delta^2}\,|x|^2\quad\text{for}\quad
 x\in R_\delta\,.
\end{equation}
By a \emph{spacelike cone} we shall mean a closed (solid) cone
in $\M$ such that all vectors going from the apex to other
points of the cone are spacelike. A~\emph{symmetrical spacelike
cone} will be the union of such cone with its reflection with
respect to its apex, and a \emph{fattened symmetrical spacelike
cone} -- the union of such cone with its reflection with
respect to a point inside the cone. An open version of any of
the defined cones will be its interior.

\setcounter{equation}{0}

\section{The model}\label{alg}

We briefly summarize the model formulated in \cite{her98}. The
choice of the test functions spaces is slightly modified.

\subsection{Electromagnetic test functions}

Let $V(s,l)$ be a real vector function of a real variable $s$ and a
future-pointing lightlike vector $l$. We shall understand
differentiability of functions $V_a$ in the sense of the action
of $L_{ab}$ and $\p_s=\p/\p s$, and denote
$\dV(s,l)=\p_sV(s,l)$.  Let~$\V_\ep$ be the real vector space
of $\C^\infty$ functions $V_a(s,l)$ which satisfy the following
additional conditions:
\begin{gather}
 V(\mu s,\mu l)=\mu^{-1}V(s,l)\,,\quad \mu>0\,,\label{hom}\\[.7ex]
 l\cdot V(s,l)=0\,,\label{ortog}\\
 |L_{b_1c_1}\ldots
 L_{b_kc_k}\dV_a(s,l)|\leq\frac{\con(t,k)}{(t\cdot
 l)^2(1+|s|/t\cdot l)^{1+\ep}}\,,\quad k\in\mN\,,\label{falloff1}\\[.5ex]
 V(+\infty,l)=-V(-\infty,l)\equiv\tfrac{1}{2}\Delta V(l)\,,\label{infty}\\[.5ex]
 L_{[ab}\Delta V_{c]}(l)=0\,,\label{DV}
\end{gather}
where the third condition holds for an arbitrarily chosen unit
timelike, future-pointing vector $t$; the bounds are then true
for any other such vector (with some other constants).
Moreover, with the use of homogeneity \eqref{hom}, the bounds
are generalized to
\begin{equation}\label{falloff2}
 |L_{b_1c_1}\ldots L_{b_kc_k}\p_s^nV_a(s,l)|
 \leq\frac{\con(t,n,k)}{(t\cdot l)^2(1+|s|/t\cdot l)^{n+\ep}}\,,\quad
 n,k\in\mN\,,
\end{equation}
It follows from the property \eqref{DV} that
\begin{equation}\label{VPhi}
 l_a\Delta V_b(l)-l_b\Delta V_a(l)=-L_{ab}\Phi_V(l)\,,
\end{equation}
where
\begin{equation}\label{Phi}
 \Phi_V(l)=-\frac{1}{4\pi}\int\frac{l\cdot \Delta V(l')}{l\cdot l'}\,d^2l'
\end{equation}
is a smooth homogeneous function. If
 $\Delta V(l)=l\al(l)$, then
\[
 \Phi_V(l)=-\frac{1}{4\pi}\int\al(l')\,d^2l'=\con\,.
\]
We also note for later use that for $v\in H_+$, there is
\begin{equation}\label{intVPhi}
 \int\frac{v\cdot\Delta V(l)}{v\cdot l}\,d^2l
 =-\int\frac{\Phi_V(l)}{(v\cdot l)^2}\,d^2l\,.
\end{equation}

The spaces $\V_\ep$ form an increasing family for $\ep\searrow
0$, so their union is a~vector space,
\begin{equation}
 \V=\bigcup_{\ep>0}\V_\ep\,.
\end{equation}
This vector space, when viewed as an Abelian group, allows the
following sub- and quotient groups:
\begin{gather}
 \V^0_\as=\{V\in\V\mid l\wedge V(s,l)=0\ \text{and}\
 \Phi_V(l)=n(2\pi/e),\ n\in\mZ\}\,,\label{V0as}\\
 L=\V/\V^0_\as\,;\label{L}
\end{gather}
the elements of the latter will be denoted by $[V]$. The space
$\V$ is equipped with a symplectic form
\begin{equation}\label{symplectic}
 \{V_1,V_2\}=\frac{1}{4\pi}\int(\dV_1\cdot V_2-\dV_2\cdot V_1)(s,l)\,ds\,d^2l\,,
\end{equation}
which is also consistently transferred to $L$.

For each $V\in\V$, the formula
\begin{equation}\label{freeV}
 A(x)=-\frac{1}{2\pi}\int \dV(x\cdot l,l)d^2l
\end{equation}
gives the Lorentz potential of a free electromagnetic field
with well-defined null asymptotes:
\begin{equation}
 \lim_{R\to\infty}RA(x\pm Rl)=\pm V(x\cdot l,l)
 -\tfrac{1}{2}\Delta V(l)
\end{equation}
and a long-range tail of electric type. This is the class of
fields which are produced in typical scattering processes
\cite{her95}. For each spacelike $x$ and any fixed $a$, the
spacelike tail is given by
\begin{equation}\label{asymptote}
 A^\as(x)=\lim_{R\to\infty}RA(a+Rx)=-\frac{1}{2\pi}
 \int \Delta V(l)\,\delta(x\cdot l)\,d^2l=A^\as(-x)\,,
\end{equation}
where $\delta$ is the Dirac measure. Let $F^\as_{ab}$ be the
electromagnetic field of this asymptotic potential. The condition
\eqref{DV} implies that $x^{}_{[a}F^\as_{bc]}(x)=0$, so this field
is of electric type. If $F^\as=0$, we shall say that the field is
infrared-regular, otherwise it will be called infrared-singular.
The symplectic form \eqref{symplectic} is a natural extension, to
the class considered here, of the usual symplectic form of free,
infrared-regular electromagnetic fields.

\subsection{Matter test functions}

 We denote by $\Sc(H_+)$ the space of smooth $4$-spinor functions on $H_+$ for
which all seminorms
\begin{equation}
 \|f\|^{H_+}_{\al,\be}=\sup |v^\al \delta^\be f(v)|
\end{equation}
are finite (with the usual multi-index notation, and supremum
over $v$ and components of the field).

For $f\in \Sc(H_+)$ the Fourier representation in the form of
the formula
\begin{equation}\label{freef}
 \psi(x)=\left(\frac{m}{2\pi}\right)^{3/2}
 \int e^{\txt-imx\cdot v\,\gamma\cdot v}\gamma\cdot v f(v)\,d\mu(v)
\end{equation}
gives a smooth Dirac field, with the timelike asymptote
determined by
\begin{equation}
 f(v)=\lim_{\la\to\infty}\la^{3/2}ie^{\txt i(m\la+\pi/4)\gamma\cdot v}\psi(\la v)\,.
\end{equation}
One has the usual scalar product in the space of these fields
\begin{equation}
 (f_1,f_2)=\int\ov{f_1(v)}\gamma\cdot v f_2(v)\,d\mu(v)=
 \int_\Sigma\ov{\psi_1}\gamma^a\psi_2(x)\,d\sigma_a(x)\,,\label{scalarpr}
\end{equation}
where the second integral is over any Cauchy surface $\Sigma$.
We denote by $\K$ the Hilbert space completion of $\Sc(H_+)$
with respect to this product.

\subsection{The algebra}\label{algebra}

The ${}^*$-algebra $\B$ of the model is generated by elements
$W([V])$, $[V]\in L$, which for simplicity will also be written as
$W(V)$, elements $\Psi(f)$, $f\in\Sc(H_+)$, and a unit $E$ by
\begin{gather}
 \begin{split}
 W(V_1)W(V_2)&= e^{\txt -\frac{i}{2}\{ V_1,V_2\}} W(V_1+V_2)\,,\\
 W(V)^* &= W(-V)\,,\ W(0)=E\,,
 \end{split}\label{weyl}\\[1ex]
[\Psi(f_1),\Psi(f_2)]_+ =0\,,\quad
[\Psi(f_1),\Psi(f_2)^*]_+ =(f_1, f_2)E\,,\label{ferm}\\[1ex]
 W(V)\Psi(f)=\Psi(S_{\Delta V}f)W(V)\,,\label{com}
\end{gather}
where
\begin{equation}\label{S}
 (S_{\Delta V} f)(v)=
 \exp\left(\dsp -\frac{ie}{4\pi}
 \int\frac{v\cdot\Delta V(l)}{v\cdot l}\,d^2l\right)\,f(v)\,.
\end{equation}
Note that the exponent function in the last formula is a multiplier in $\Sc(H_+)$,
so the operator $S_{\Delta V}$ is a~linear automorphism of
$\Sc(H_+)$. This can be easily seen: since for $t\cdot l=1$ and
$v\in H_+$ there is
 $|v\cdot l|^{-1}<|v^0|+|\vec{\vspace{1pt}v}|$, so $\left|\dsp
 \int\frac{\Delta V^a(l)l^{\al}}{(v\cdot l)^{|\al|+1}}\,d^2l\,\right|$ is
polynomially bounded for any multi-index $\al$.  Note also that, by the identity \eqref{intVPhi} and
definitions \eqref{V0as} and \eqref{L}, there is $S_{\Delta
V_2}=S_{\Delta V_1}$ for $V_2\in[V_1]\in L$, so the algebra is
properly defined.

The elements $\Psi(f)$ generate a subalgebra $\B^+$ of the
\rm{CAR} type, and the elements $W(V)$ -- a subalgebra $\B^-$ of
the \rm{CCR} type. We denote by $\beta_V$ the automorphisms of
$\B^+$ defined by
\begin{equation}
 \beta_V(C)=W(V)CW(-V)\,,
\end{equation}
forming a group, $\beta_{V_1}\beta_{V_2}=\beta_{V_1+V_2}$.

Regular, translationally covariant, positive energy
representations of $\B$ are shown, up to a unitary equivalence, to
form a class defined in the following way. Let $\pi_F$ be the
standard positive energy Fock representation of $\B^+$ on $\Hc_F$
with the Fock vacuum vector $\W_F$, and $\pi_r$ be any regular,
translationally covariant, positive energy representation of
$\B^-$ on $\Hc_r$. Define operators $\pi(A)$ on
$\Hc=\Hc_F\otimes\Hc_r$ by
\begin{equation}\label{rep}
 \begin{split}
    \pi(C)&=\pi_F(C)\otimes\id_r\,,\quad C\in\B_\as^+\,,\\
    \pi(W(V))[\pi_F(B)\W_F\otimes\varphi]&=
    \pi_F(\be_VB)\W_F\otimes\pi_r(W(V))\varphi\,,\quad B\in\B_\as^+\,.
 \end{split}
\end{equation}
Then $\pi$ extends to a regular, translationally covariant
positive energy representation of $\B$. We add one further
demand to our selection criterion, that
$\pi_r(W(V_1))=\pi_r(W(V_2))$ whenever $l\wedge V_1=l\wedge
V_2$, which is related to the gauge invariance.

One shows that all representations from the class thus defined
determine the same $C^*$-norm on $\B$; the completion of
$\B$ in this norm is the $C^*$-algebra $\F$ of the model.

\setcounter{equation}{0}

\section{Spacelike localization of electromagnetic
fields}\label{locelm}

We now want to equip the elements of the algebra with spacetime
localization properties. We start with the electromagnetic
fields, which have direct observable status. The way to ascribe
spacetime properties to them is to represent the classical test
fields $A$ in \eqref{freeV} as
\begin{equation}\label{freeJ}
 A(x)=4\pi\int D(x-y)\,J(y)\,d^4y\,.
\end{equation}
Here $J$ is a classical conserved test current field, and
$D(x)=D(0,x)$, with
\begin{equation}
 D(m,x)=\frac{i}{(2\pi)^3}\int\sgn p^0\delta(p^2-m^2)e^{-ip\cdot x}\,dp\,.
\end{equation}
We want the supports of $J$ to be contained between two Cauchy
surfaces. This may be interpreted as a generalized time-slice
property.

We shall be concerned with conserved test currents $J$ which
are elements of $\Sc^3_{3+\ep}(R_{\ga,\delta})$. Then the
asymptote $J_\as$ has the support in $R_\delta$. For such
currents the integral in \eqref{freeJ} is absolutely convergent
and determines a~corresponding $A$. We want to find out whether
this potential is of the type given by \eqref{freeV}. We start
with a useful subsidiary result.
\begin{lem}
 Let $J_\as$ be a homogeneous of degree $-3$ vector function, smooth
 outside the origin, with support in $R_\delta$.
 The following statements are equivalent.\\
\hspace*{1em}(i) The continuity equation
\begin{equation}\label{ascont}
 \p\cdot J_\as(x)=0
\end{equation}
is satisfied distributionally.\\
\hspace*{1em}(ii) $J_\as$ satisfies the following conditions on
$H_-$
\begin{gather}
 \delta\cdot J_\as(z)+3z\cdot J_\as(z)=0\,,\label{contdif}\\
 \int z\cdot J_\as(z)\,d\nu(z)=0\,.\label{contzero}
\end{gather}
\hspace*{1em}(iii) $J_\as$ is an asymptote of some conserved
current $J\in\Sc^3_{3+\ep}(R_{\ga,\delta})$.

In particular, these conditions are satisfied for $J_\as$ of
the special form
\begin{equation}\label{zg}
 J_\as(x)=xg(x)\quad \txt{with}\quad
 \int g(z)\,d\nu(z)=0\,,
\end{equation}
where $g$ is a scalar function homogeneous of degree $-4$, smooth outside the origin.
\end{lem}
\begin{proof}
The condition \eqref{contdif} is equivalent to \eqref{ascont}
for $x$ outside the origin (use \eqref{zd}). If it holds, then
we have for any test function $\vp$
\begin{multline}
 \int J_\as^b(x)\p_b\vp(x)\,d^4x=
 \lim_{\epsilon\to
 0}\int_{x^2=-\epsilon^2}\vp(x)J_\as^b(x)\,d\sigma_b(x)\\
 =\vp(0)\int z\cdot J_\as(z)\,d\nu(z)\,,
\end{multline}
which proves the equivalence of (i) and (ii). Let $\rho$ be a
smooth function with support in $|x|\leq\ga/\sqrt{2}$ and such
that \mbox{$\int\rho(x)\,d^4x=1$}. The vector function
\begin{equation}\label{Jrho}
 J_\rho=\rho*J_\as
\end{equation}
is easily shown to be in $\Sc^3_{3+\ep}(R_{\ga,\delta})$ with
the asymptote $J_\as$, and if (i) is true, then it satisfies
the continuity equation. Conversely, if $J_\as$ is the
asymptote of a~conserved $J\in\Sc^3_{3+\ep}(R_{\ga,\delta})$,
then it is supported in $R_\delta$ and \eqref{contdif} is the
limit of the continuity equation $\p\cdot J(x)=0$ for
$x^2\to-\infty$. Integrating the latter equation over the
region $x^2\geq-R^2$ and taking the limit $R\to\infty$ one
arrives at \eqref{contzero}. The statement concerning
\eqref{zg} is easily checked.
\end{proof}

We note for future use that by \eqref{zsto} and \eqref{contdif}
one has for any continuously differentiable function $f(z)$
\begin{equation}\label{Jdot}
 \int J_\as\cdot\delta f(z)\,d\nu(z)=0\,.
\end{equation}
\begin{thm}
 Let $J\in\Sc^3_{3+\ep}(R_{\ga,\delta})$ be a conserved current.
 Then the function
\begin{equation}\label{VR}
 \dV(s,l)=\frac{1}{s}\big(V_0(s,l)-V_0(0,l)\big)\,,
\end{equation}
where
\begin{equation}\label{V0}
 V_0(s,l)=\int\delta(s-x\cdot l)H_3J(x)\,d^4x\,,
\end{equation}
satisfies conditions \eqref{hom} and \eqref{ortog}, and $J$ and
$V$ generate the same $A$ according to \eqref{freeJ} and
\eqref{freeV} respectively. If the asymptote of $J$ is odd:
\begin{equation}
 J_\as(-x)=-J_\as(x)\,,
\end{equation}
then $V_0(0,l)=0$, so $V$ satisfies also \eqref{falloff1}, and
it may be then obtained by
\begin{equation}\label{VlimR}
 V(s,l)=\lim_{R\to\infty}V^R(s,l)\,,\quad
 V^R(s,l)=\int_{x^2\geq-R^2}\delta(s-x\cdot l)J(x)\,d^4x
\end{equation}
 with $V^R(s,l)$ uniformly bounded and with
\begin{equation}\label{DeltaV}
 \Delta V(l)=\int\frac{J_\as(z)}{z\cdot l}\,d\nu(z)
\end{equation}
(the integral in the principal value sense). If in addition
$L_{[ab}\Delta V_{c]}(l)=0$, then $V\in\V_\ep$. This is, in
particular, fulfilled for $J_\as$ of the type given by \eqref{zg}
with even~$g(z)$.

If $J_1$ and $J_2$ are two currents satisfying all the above
assumptions, then
\begin{equation}\label{sympl}
 \{V_1,V_2\}=\lim_{R\to\infty}\tfrac{1}{2}
 \int_{x^2\geq-R^2}[J_1\cdot A_2-J_2\cdot A_2](x)\,d^4x\,.
\end{equation}
\end{thm}
\begin{proof}
We first observe that as $H_3J(x)$ vanishes as $|x|^{-3-\ep}$ in
infinity, the integral \eqref{V0} is absolutely convergent, and
relations \eqref{hom} and \eqref{ortog} are easily seen to hold for $V_0$. Moreover,
with $X_{ab}=x_a\p/\p x^b-x_b\p/\p x^a$, we have
\begin{multline}\label{estV0}
 |L_{a_1b_1}\ldots L_{a_kb_k}V_0(s,l)|=\left|\int\delta(s-x\cdot l)
 X_{a_1b_1}\ldots X_{a_kb_k}H_3J(x)\,d^4x\right|\\
 \leq\con\int\delta(s-x\cdot l)(1+|x|)^{-3-\ep}\,d^4x
 \leq\frac{\con}{t\cdot l(1+|s|/t\cdot l)^{\ep}}\,.
\end{multline}

If $A$ is generated by $J$, then one finds easily that $H_1A$ is
generated by $H_3J$. It is then also easily seen, using the
representation
\[
 D(x)=-(1/8\pi^2)\int\delta'(x\cdot l)\,d^2l\,,
\]
that $\dV_0$ generates $H_1A$ by \eqref{freeV}. But then it
follows that $A$ may be obtained by \eqref{freeV} from $\dV$
defined by \eqref{VR}.

We want to obtain another form of $V_0$. For any $R>0$ we have
\begin{multline}
 \p\cdot\bigg\{x\,\delta(s-x\cdot l)
 \Big[J(x)-\theta(-x^2\!-R^2)J_\as(x)\Big]\bigg\}\\
 =-s\,\delta'(s-x\cdot l)
 \Big[J(x)-\theta(-x^2\!-R^2)J_\as(x)\Big]\\
 +\delta(s-x\cdot l)H_3J(x)
 -2R^2\delta(s-x\cdot l)\delta(x^2+R^2)J_\as(x)\,.
\end{multline}
The l.h.s.\ yields zero when integrated over whole space, so we
find
\begin{multline}\label{V0R}
 V_0(s,l)=s\,\p_s\int\delta(s-x\cdot l)
 \Big[J(x)-\theta(-x^2\!-R^2)J_\as(x)\Big]\,d^4x\\
 +\int\delta\Big(\frac{s}{R}-z\cdot l\Big)J_\as(z)\,d\nu(z)\,.
\end{multline}
Setting here $s=0$, we find
\begin{equation}
 V_0(0,l)=\int\delta(z\cdot l)J_\as(z)\,d\nu(z)\,,
\end{equation}
so if $J_\as$ is odd, what we assume from now on, there is
$V_0(0,l)=0$, and then $V$ satisfies the bounds \eqref{falloff1} (use \eqref{estV0}).
We note that if $V_0(0,l)\neq0$, then $\dV(s,l)$ falls off only as
$1/|s|$ and is outside the class~$\V$.

We integrate \eqref{VR} with the use of \eqref{V0R}, and find
\begin{multline}\label{VsR}
 V(s,l)-V(-\infty,l)=V^R(s,l)
 +\int_{-\infty}^{s/R}\frac{1}{\tau}
 \bigg\{\int\delta(\tau-z\cdot l)
 J_\as(z)\,d\nu(z)\bigg\}\,d\tau\\
 +\int_{x^\leq-R^2}\delta(s-x\cdot l)
 (J-J_\as)(x)\,d^4x\,,
\end{multline}
with $V^R$ as defined in \eqref{VlimR}. The last term vanishes
both in the limit \mbox{$R\to\infty$} as well as $|s|\to\infty$,
and $V^R(s,l)$ vanishes for $|s|\to\infty$; the uniform
boundedness of $V^R(s,l)$ is also easily seen. We write down the
limit versions of \eqref{VsR} for $R\to\infty$ and for $s\to\infty$,
respectively (remember that
$V(+\infty,l)=-V(-\infty,l)=\tfrac{1}{2}\Delta V(l)$)
\begin{equation}\label{VDeltaV}
 V(s,l)+\tfrac{1}{2}\Delta V(l)=\lim_{R\to\infty}V^R(s,l)
 +\int_{-\infty}^0\frac{1}{\tau}
 \bigg\{\int\delta(\tau-z\cdot l)
 J_\as(z)\,d\nu(z)\bigg\}\,d\tau\,,
\end{equation}
\begin{equation}\label{DeltaV'}
 \Delta V(l)=\int_{-\infty}^{+\infty}
 \frac{1}{\tau}
 \bigg\{\int\delta(\tau-z\cdot l)
 J_\as(z)\,d\nu(z)\bigg\}\,d\tau
 =\lim_{\ep\to0}\int_{|z\cdot l|\geq\ep}
 \frac{J_\as(z)}{z\cdot l}\,d\nu(z)\,.
\end{equation}
The last equation gives \eqref{DeltaV}. Due to the oddness of
$J_\as$ the second term on the r.h.s.\ of \eqref{VDeltaV} is
then $\tfrac{1}{2}\Delta V(l)$, and we thus obtain
\eqref{VlimR}. If \eqref{DV} is satisfied, then $V\in\V_\ep$. We
note that the differentiation on the cone is transferred to the
differentiation on the hyperboloid:
\begin{equation}
 L_{ab}\int\frac{J_\as(z)}{z\cdot l}\,d\nu(z)=
 \int\frac{(z_a\delta_b-z_b\delta_a)J_\as(z)}{z\cdot l}\,d\nu(z)\,,
\end{equation}
therefore $\Delta V$ is smooth, and for $J_\as=zg(z)$ the
condition \eqref{DV} is satisfied automatically.

The last point concerns the symplectic form. We have
\begin{multline}
 \frac{1}{4\pi}\int(\dV_1\cdot V^R_2-\dV_2\cdot
 V^R_1)(s,l)\,ds\,d^2l\\
 =\frac{1}{4\pi}
 \int_{x^2\geq-R^2}\Big[\dV_1(x\cdot l)J_2(x)-
 \dV_2(x\cdot l)\cdot J_1(x)\Big]\,d^2l\,d^4x\\
 =\frac{1}{2}\int_{x^2\geq-R^2}(J_1\cdot A_2-J_2\cdot A_1)(x)\,,d^4x
\end{multline}
due to the representation \eqref{freeV}. As $V^R_i(s,l)$ are
uniformly bounded, by the Lebesgue theorem the l.h.s.\ has a
finite limit $\{V_1,V_2\}$ for $R\to\infty$, so also the
r.h.s.\ has a finite limit, and one arrives at \eqref{sympl}. We
note, however, that the integrand of the r.h.s.\ is not
absolutely integrable on the whole space. The mechanism of the
convergence in the limit relays on the fact that the asymptotes
of $J_i$ are odd, while those of $A_i$ are even, so their
products do not contribute, if integration is done in the above
sense.
\end{proof}

A particular test current
$J_\rho\in\Sc^3_{3+\ep}(R_{\ga,\delta})$ with the given
asymptote $J_\as$ supported in $R_\delta$ was given in
\eqref{Jrho}. We want to find its corresponding function
$V_\rho$. We start with the following geometrical observation:
for $y\in R_\delta$ and $|x-y|\leq\ga$ there is
\begin{equation}\label{theta}
 |\theta(x^2+R^2)-\theta(y^2+R^2)|
 \leq\theta\big(-y^2-(R-R_1)^2\big)
 \theta\big(y^2+(R+R_2)^2\big)
\end{equation}
for $R\geq R_1$, with some $\ga$- and $\delta$-dependent
constants $R_1$, $R_2$. This seems rather intuitive, but we
give a formal proof in Appendix. It is then easy to see that
instead of formula \eqref{VlimR} one can use
$V_\rho=\lim_{R\to\infty}V^{'R}_\rho$ with
\begin{equation}
 V^{'R}_\rho(s,l)=\int\delta(s-w\cdot l-y\cdot l)\rho(w)
 \theta(y^2+R^2)J_\as(y)\,d^4w\,d^4y\,.
\end{equation}
If we denote
\begin{gather}
 H(s,l)=\int\sgn(s-x\cdot l)\rho(x)\,d^4x\,,\\
 V_\as^R(s,l)=\int\delta(s-x\cdot l)\theta(x^2+R^2)J_\as(x)\,d^4x\,,
\end{gather}
then we have
\begin{equation}\label{VsR'}
 V^{'R}_\rho(s,l)=\tfrac{1}{2}\int \dot{H}(s-\tau,l)V_\as^R(\tau,l)\,d\tau\,.
\end{equation}
Using in the following first step \eqref{zi} and the homogeneity
of $J_\as(x)$, and in the second step oddness of $J_\as(x)$, we
find
\begin{multline}
 V_\as^R(\tau,l)=\int\theta\left(\frac{\tau}{z\cdot l}\right)
 \theta\left(R-\frac{\tau}{z\cdot l}\right)
 \frac{J_\as(z)}{|z\cdot l|}\,d\nu(z)\\
 = \tfrac{1}{2}\sgn(\tau)
 \int_{|z\cdot l|\geq\frac{|\tau|}{R}}\frac{J_\as(z)}{z\cdot l}\,d\nu(z)\,.
\end{multline}
Thus for  $R\to\infty$ the absolute value of \eqref{VsR'} remains
bounded, and one finds
\begin{equation}\label{VrhoH}
 V_\rho(s,l)=H(s,l)\tfrac{1}{2}\Delta V(l)\,,
\end{equation}
with $\Delta V(l)$ given by \eqref{DeltaV}. Note that
$H(\pm\infty,l)=\pm1$.

Assume now that  $\Delta V(l)$ satisfies \eqref{DV} and is
therefore determined up to a gauge by $\Phi_V(l)$.
\begin{prop}\label{PhiJ}
 For $\Delta V(l)$ given by \eqref{DeltaV} there is
\begin{equation}
 \Phi_V(l)=\int z\cdot J_\as(z)\log|z\cdot l|\,d\nu(z)\,.
\end{equation}
\end{prop}
\begin{proof}
We observe that the formula \eqref{Phi} defines in fact a
continuous, homogeneous function $\Phi_V(x)$ for $x$ in the closed
future lightcone. For $x$ inside the cone, and with
$v=x/\sqrt{x^2}$, one finds
\begin{equation}\label{Phix}
 \Phi_V(x)=-\int\frac{v\cdot J_\as(z)}{\sqrt{(v\cdot z)^2+1}}
 \log\Big[\sqrt{(v\cdot z)^2+1}+v\cdot z\Big]\, d\nu(z)\,,
\end{equation}
where we used the following formula valid for $v^2=1$, $v^0>0$
and $z^2=-1$:
\begin{equation}
 \int\frac{d^2l}{v\cdot l\,z\cdot l}=
 \frac{4\pi}{\sqrt{(v\cdot z)^2+1}}
 \log\Big[\sqrt{(v\cdot z)^2+1}+v\cdot z\Big]\,.
\end{equation}
We observe that $\delta_a^{(z)}(v\cdot z)=v_a+v\cdot z\, z_a$,
which allows us to write the integrand in \eqref{Phix} as
\begin{multline}
 \tfrac{1}{2}J_\as(z)\cdot\delta\left(\log\Big[\sqrt{(v\cdot z)^2+1}+v\cdot z\Big]\right)^2\\
 +z\cdot J_\as(z)\bigg[1-\frac{|v\cdot z|}{\sqrt{(v\cdot
 z)^2+1}}\bigg]\log\Big[\sqrt{(v\cdot z)^2+1}+|v\cdot z|\Big]\\
 -z\cdot J_\as(z)\log\tfrac{1}{2}\Big[\sqrt{(x\cdot z)^2+x^2}+|x\cdot z|\Big]
 +z\cdot J_\as(z)\log\tfrac{1}{2}\sqrt{x^2}\,,
\end{multline}
where we used the fact that
$\xi\log\big[\sqrt{\xi^2+1}+\xi\big]=|\xi|\log\big[\sqrt{\xi^2+1}+|\xi|\big]$.
The first and the last terms give no contribution to the
integral (use \eqref{Jdot} and \eqref{contzero} respectively).
We consider the other terms in the limit $x\to l$. In this
limit \mbox{$|v\cdot z|$} tends to $+\infty$ almost everywhere,
and the second term remains bounded by \mbox{$\con|z\cdot
J_\as(z)|$} and tends to zero almost everywhere, so the
contribution to the integral vanishes in this limit. Finally,
the third term gives the thesis.
\end{proof}

The above result has an interesting consequence.
\begin{prop}\label{Jzg}
 Let $J\in\Sc^3_{3+\ep}(R_{\ga,\delta})$ be a conserved current
 with an odd asymptote $J_\as$ and the corresponding function $V(s,l)$.
 Let  $L_{[ab}\Delta V_{c]}(l)=0$, so that $V\in\V_\ep$.

 Then there exists a current $J'$ of the same type, but whose
 asymptote is of the particular form $J'_\as(z)=zg(z)$, such
 that the corresponding function $V'(s,l)$ satisfies
 \begin{equation}
  l\wedge (V'\fm V)(s,l)=0\quad\quad\text{and}\quad\quad
  \Phi_{V'}(l)=\Phi_V(l)\,.
 \end{equation}
 Thus, in particular, $V'(s,l)-V(s,l)\in\V^0_\as$ and
 $[V']=[V]\in L$.
\end{prop}
\begin{proof}
We set $J'=J+\rho*(J'_\as-J_\as)$, where $J'_\as$ is
homogeneous of degree $-3$ and on the unit hyperboloid given by
$J'_\as(z)=-z\,z\cdot J_\as(z)$, which then indeed is the
asymptote of $J'$. Then by \eqref{VrhoH} there is
\begin{equation}
 (V'\fm V)(s,l)=H(s,l)\tfrac{1}{2}(\Delta V'\fm \Delta V)(l)
\end{equation}
  and by
Proposition~\ref{PhiJ}: $\Phi_{V'}(l)=\Phi_V(l)$. Therefore
$l\wedge(\Delta V'\fm\Delta V)(l)=0$, which completes the proof.
\end{proof}

The net result of the present section to this point is the
identification of a~class of currents giving rise to test
elements $[V]\in L$ of our electromagnetic Weyl algebra. Now we
want to show that the whole group $L$ is covered in this way,
and even more, that the class may be still narrowed. We start
with an auxiliary result.
\begin{lem}
Let a smooth function $W(s,l)$ be homogeneous of degree $n-2$,
$W(\mu s,\mu l)=\mu^{n-2}W(s,l)$ ($\mu>0$), and satisfy the falloff
conditions
\begin{equation}
 |L_{b_1c_1}\ldots
 L_{b_kc_k}W(s,l)|\leq\con(k)\,\frac{(t\cdot l)^{n-2}}
 {(1+|s|/t\cdot l)^\ep}\,,\quad k\in\mN\,.
\end{equation}
Denote $W^{(k)}(s,l)=\p_s^kW(s,l)$ and set
\begin{equation}
 K(x)=-\frac{1}{2\pi}\int W^{(n)}(x\cdot l,l)\,d^2l\,.
\end{equation}
Then for each fixed $\delta\in(0,1)$ one has in the region
$R_\delta$ the bounds
\begin{equation}
 |K(a+x)|\leq\con(\delta)\,(1+|x|)^{-n-\ep}\,.
\end{equation}
\end{lem}
\begin{proof}
It is sufficient to show this for $a=0$, as the properties of
$W$ are conserved under translations. For $n=0$ and $x\in R_\delta$,  we have
\begin{equation*}
 |K(x)|\leq\con\int_{-1}^1
 \frac{du}{\big(1+\big||x^0|+|\vec{\vspace{1pt}x}|u\big|\big)^\ep}
 \leq\con(\delta)\,(1+|x|)^{-\ep}\,,
\end{equation*}
We proceed by induction with respect to $n$. If we denote
\[
 \tilde{x}(t,l)=(t\cdot l)^{-1}x+(t\cdot l)^{-2}t\cdot x\,l\,,
 \]
then we have the identity
\begin{multline}
 L_{ab}\Big[t^a\tilde{x}^bW^{(n-1)}(x\cdot l,l)\Big]\\
 = x^2W^{(n)}(x\cdot l,l)
 + \left[t^a\tilde{x}^b L'_{ab}
 +\frac{x\cdot l}{(t\cdot l)^2}\right]W^{(n-1)}(x\cdot l,l)\,,
\end{multline}
where $L'_{ab}W^{(n-1)}(x\cdot
l,l)=L_{ab}W^{(n-1)}(s,l)|_{s=x\cdot l}$. The integral of the
l.h.s.\ over $l$ vanishes, so by induction we have
\begin{multline}
 |K(x)|\leq\min\left\{\con,\con(\delta)\frac{|x|}{|x^2|}(1+|x|)^{-n+1-\ep}\right\}\\
 \leq
 \con(\delta)(1+|x|)^{-n-\ep}\,.
\end{multline}
\end{proof}

We can now prove our main result of this subsection.
\begin{thm}\label{AfromJ}
 Let $A$ be given by the formula \eqref{freeV} with
 $V\in\V_\ep$, and chose an arbitrary set of the type $R_{\ga,\delta}$.
 Then:\\
 \hspace*{1em}(i) There exists $V'\in\V_\ep$ such that
$[V']=[V]$ and the corresponding potential $A'$ may be
represented as a radiation potential of a test current \linebreak
$J'\in\Sc^3_{3+\ep}(R_{\ga,\delta})$ with the asymptote of the form $J'_\as(x)=x\rho(x)$, with \linebreak
\mbox{$\rho(-x)=\rho(x)$}, supported in~$R_\delta$.\\
 \hspace*{1em}(ii) The test current $J'$ may be represented
as a sum of currents with the same properties, but in addition
each of the currents is supported in a fattened symmetrical
spacelike cone contained in $R_{\ga,\delta}$. For each cover of
the set $R_{\ga,\delta}$ with such cones there is a
corresponding split of $J'$.
\end{thm}
\begin{proof}
For a given $A$ and $V$, we define
\begin{equation}\label{C}
 C^a(x)=-\frac{1}{2\pi}\int\frac{V^a(x\cdot l,l)}{t\cdot l}\,d^2l\,,\quad
 B^{ab}=C^at^b-C^bt^a\,.
\end{equation}
Then $\Box B^{ab}(x)=0$ and $A^a(x)=\p_bB^{ab}(x)$. Moreover,
with the use of the above lemma one finds easily that for $x\in R_\delta$, there is
\begin{equation}\label{falloffC}
 |D^\al H_0C(a+x)|\leq\con(a,\delta,\al)(1+|x|)^{-|\al|-\ep}\,.
\end{equation}

Let now $F$ be a smooth function on the spacetime which for
$|x|\geq\ga$,
for some $\ga>0$, satisfies:\\
\hspace*{1em} (i) $F(\mu x)=F(x)$ for all $\mu\geq1$ (homogeneity),\\
\hspace*{1em} (ii) $F(-x)=-F(x)$,\\
\hspace*{1em} (iii) $F(x)=1/2$ for $x^0\geq \delta|\vec{\vspace{1pt}x}|$ for
some $\delta\in(0,1)$.\\
Note that the supports of derivatives of $F$ are contained in
$R_{\ga,\delta}$.  We claim that
\begin{equation}
 B^{ab}(x)=4\pi\int D(x-y)\vp^{ab}(y)\,d^4y\,,
\end{equation}
where $\vp^{ab}(y)=\Box(F(y)B^{ab}(y))$. Indeed, the support of
$\vp$ is contained in $R_{\ga,\delta}$, and for $x$ in the
future of $R_{\ga,\delta}$ the r.h.s. may be written as
\begin{equation*}
 4\pi\int D_\mathrm{ret}(x-y)\Box\Big([F(y)+\tfrac{1}{2}]B^{ab}(y)\Big)\,d^4y\,,
\end{equation*}
which yields the l.h.s. upon integration by parts. But both
sides satisfy the wave equation, so the equality holds
everywhere.

The fall-off properties \eqref{falloffC} now easily imply that
$\vp\in\Sc^2_{2+\ep}$. Moreover, the support of $\vp$ is contained
in $R_{\ga,\delta}$ and that of the asymptote $\vp_\as$ in
$R_\delta$, and the asymptote is even: $\vp_\as(-x)=\vp_\as(x)$.
The potential $A$ has now the representation \eqref{freeJ} with
the test current $J^a=\p_b\vp^{ab}$, which is an element of
$S^3_{3+\ep}$, has similar support properties as $\vp$, and its
asymptote is odd: $J_\as(-x)=-J_\as(x)$. Thus $J$~satisfies all
the assumptions of the Proposition~\ref{Jzg}. The current $J'$
defined in the proof of this proposition may be written in the
present case as $J'=J'_\reg+J'_\sing$ with
\begin{gather}
 J^{'a}_\reg=\p_b\vp^{ab}_\reg\,,\quad
 \vp_\reg=\vp-\rho*\vp_\as\,,\label{J'reg}\\
 J'_\sing=\rho*J'_\as\,,\quad
 J'_\as(x)=x\,\bigg(\frac{x_c\p_b\vp_\as^{cb}(x)}{x^2}\bigg)\,.\label{J'sing}
\end{gather}
This completes the proof of (i).

To show (ii), we apply the above construction to
$R_{\ga',\delta'}$ with $\ga'<\ga$, $\delta'<\delta$, and note that
the two parts $J'_\reg$ and $J'_\sing$ may be considered
separately. For the first part we note a rather obvious fact: for
each cover of $R_{\ga',\delta'}$ with open fattened symmetrical
spacelike cones contained in $R_{\ga,\delta}$ there exist a
decomposition of unity on $R_{\ga',\delta'}$ with smooth functions
$f_k$ supported in the respective fattened symmetrical cones, taking values
in $\<0,1\>$ and with bounded all derivatives. The currents
\mbox{$J^{'a}_{\reg,k}=\p_b(f_k\vp_\reg^{ab})$} satisfy the
thesis. For the second part we note that the intersection of $H_-$
with $R_{\delta'}$ may be covered by arbitrarily small symmetrical
patches, which are open as subsets of $H_-$ and are contained in
$R_\delta$. For each such cover there exists a corresponding
decomposition of unity on $R_{\delta'}\cap H_-$ with smooth, even
functions $g_k(z)$ supported in the respective patches, taking
values in $\<0,1\>$ and with bounded derivatives. We extend these
functions by homogeneity and define
\begin{equation}
 J'_{\sing,k}=\rho*J'_{\as,k}\,,\quad
 J'_{\as,k}(x)=x\,
 \bigg(\frac{x_c\p_b\vp_{\as,k}^{cb}(x)}{x^2}\bigg)\,,\quad
 \vp_{\as,k}^{ab}=g_k\vp_\as^{ab}\,.\label{Jk'sing}
\end{equation}
The asymptotes $J'_{\as,k}$ are odd and satisfy
\begin{equation}
 \int z\cdot J'_{\as,k}(z)\,d\nu(z)=
 \int\delta_b\left(z_c\vp_{\as,k}^{cb}(z)\right)\,d\nu(z)=0
\end{equation}
by \eqref{zsto}, so $J'_{\sing,k}$ are conserved currents by
\eqref{contzero}. Their sum yields $J'_\sing$, which ends the
proof.
\end{proof}

\setcounter{equation}{0}

\section{Localization of Dirac fields and
observables}\label{locdir}

Fields carrying charge do not represent observables. Even more, in
full electrodynamics they undergo local gauge transformations,
thus to form an observable with the use of them one has to
compensate not only the global, but also local gauge scaling. If
$\Psi(x)$ and $A(x)$ represent `local quantum spacetime fields',
then a way to achieve this is to give a precise meaning (by
smearing, renormalization etc.) to the heuristically formed
quantities
 $\ov{\Psi}(x)\exp\bigg(-ie\int\limits_x^yA(z)dz\bigg)\Psi(y)$.
Localization of this quantity, if it can be defined, should be
determined by spacetime points $x$ and $y$ and the integration
path between them.

Single fields creating or annihilating a physical charged
particle, on the other hand, interpolate between different
representations of observables. However, because of the Gauss
law they cannot be local. Staying at the adopted heuristic
level, the best that one can do is to cut the above quantity in
two and obtain
 $\exp\bigg(-ie\int\limits_\infty^yA(z)dz\bigg)\Psi(y)$,
where the path goes to spacelike infinity. The expectation then
would be that the effect of this operation is invisible in the
region spacelike to the localization of the integration path.

The above naive picture has its more refined counterpart in the
algebraic analysis of the superselection sectors in quantum
electrodynamics made by Buchholz~\cite{bu82}. The idea behind the
selection criterion adopted in this analysis is that by an
appropriate choice of the `radiation cloud' superimposed on a
charged state one can concentrate at a given time the electric
flux at spacelike infinity in an arbitrarily chosen patch on the
2-sphere in the infinity of 3-space. The causal influence of the
presence of the charge in this state may be thus made to vanish in
the causal complement of some spacelike cone in Minkowski space.

We shall now investigate this question in the model defined here.
Our algebra is an algebra of fields, not only observables, thus we
formulate the problem in their terms. We shall ask whether, in
representations defined in Section \ref{algebra}, by composing the
charged field $\pi(\Psi(f))$ with some radiation cloud and a
subsequent rescaling (to push the cloud to spacelike infinity), one
can obtain a modified field restricted to a fattened symmetrical
spacelike cone. The infrared tails are symmetric in the class of
fields considered in the model, thus the replacement of spacelike
cones by fattened symmetrical spacelike cones is unavoidable.

We shall see that the answer to this question is negative for a
rather general construction reflecting in an obvious way the
above idea. This seems to disagree
also with expectations based on perturbative calculations in QED.
The `perturbative axiomatic' construction of the physical state
space by Steinmann \cite{ste} may be seen as the strongest
indication in this direction. We postpone the discussion of this
point to the concluding section.

On the other hand, the same construction will allow us to construct
local observables formed as products of `dressed' Dirac fields
and their adjoints.

\subsection{Spacelike test functions}

To ascribe localization to elements $\Psi(f)$, we first have to
interpret test functions in spacetime terms; this will be done in
this subsection. However, this will not give the full answer to
the question because of noncommutativity with observables $W(V)$.
We treat then the addition of the clouds in further subsections.

The first step is achieved, in analogy to the electromagnetic
case, by representing the classical test field $\psi$ in
\eqref{freef} as
\begin{equation}\label{freeX}
 \psi(x)=\frac{1}{i}\int S(m,x-y)\chi(y)\,d^4y\,,
\end{equation}
where $\chi$ is a classical test 4-spinor field and
$S(m,x)=(i\gamma\cdot\p+m)D(m,x)$. We want the support $\chi$
to be contained between two Cauchy surfaces.

It is easy to show that the Fourier representation of $S(m,x)$
may be written as
\begin{equation}
 S(m,x)=i\left(\frac{m}{2\pi}\right)^3\int
 e^{\txt-imx\cdot v\,\gamma\cdot v}\gamma\cdot v\,d\mu(v)
\end{equation}
and then the Fourier connection between $f(v)$ and $\chi(x)$ in
the integral representations of the Dirac field $\psi$ given
respectively by \eqref{freef} and \eqref{freeX} takes the form
\begin{equation}\label{fchi}
 f(v)=\left(\frac{m}{2\pi}\right)^{3/2}\int
 e^{\txt imv\cdot x\,\gamma\cdot v}\chi(x)\,d^4x\,.
\end{equation}
It is clear that if $\chi\in\Sc(\M)$, the Schwartz functions
space, then $f\in\Sc(H_+)$. For the converse statement we note
first the following analogue of the `regular wave packet'
property.
\begin{prop}
 If $f\in\Sc(H_+)$, then for each $\delta\in(0,1)$ the Dirac
 field $\psi$ formed by \eqref{freef} satisfies in the region $R_\delta$ the bounds
\begin{equation}
 |D^\beta \psi(x)|\leq\con(\delta,|\beta|,n)(1+|x|)^{-n}
\end{equation}
for each $\beta$ and each $n\in\mN$.
\end{prop}
\begin{proof}
The representation \eqref{freef} is proportional to the sum of
two terms $\int e^{\mp imv\cdot x}f_\pm(v)\,d\mu(v)$ with
$f_\pm=P_\pm(v)f(v)$, $P_\pm(v)=\tfrac{1}{2}(1\pm \ga\cdot v)$.
It is clear that application of $D^\beta$ only modifies
functions $f_\pm$. Now, for any $g\in\Sc(H_+)$ and $x^2<0$, we have the identity
\begin{multline}
 \int e^{\pm imv\cdot x}g(v)\,d\mu(v)\\=
 \Big(\frac{\pm i}{m}\Big)^n
 \int \frac{e^{\pm imv\cdot x}}{[x^2-(v\cdot x)^2]^n}
 \left[\,\prod_{k=1}^n
 x\cdot\Big(\delta+(2k-3)v\Big)\right]g(v)\,d\mu(v)\,,
\end{multline}
where the operators under the product sign are ordered from
right to left with increasing $k$. This is easily shown by
induction with respect to $n$ (integrate the r.h.s.\ by parts
with the use of \eqref{vsto}). But using \eqref{loreu} we have
 $|x^2-(v\cdot x)^2|\geq\con(\delta)\,|x|^2$ for $x\in R_\delta$.
 This leads easily to the thesis.
\end{proof}

\begin{thm}
 Let $\psi$ be given by the formula \eqref{freef} with
 $f\in\Sc(H_+)$, and chose an arbitrary set of the type $R_{\ga,\delta}$
 (here $\delta=0$ is also admitted).
 Then there exists $\chi\in\Sc(R_{\ga,\delta})$ which generates
 $\psi$ by \eqref{freeX} (and, therefore, generates $f$ by \eqref{fchi}).
\end{thm}
\begin{proof}
 Let $F$ be the function defined in the proof of Theorem
 \ref{AfromJ}, and set \mbox{$\chi=(\ga\cdot\p+im)(F\psi)$}. This
 function has support in $R_{\ga,\delta}$, and with the use of the last proposition one then easily shows that it is a Schwartz function. Using  the method employed in the proof of Thm.\,\ref{AfromJ}, one finds that $\chi$ generates~$\psi$.
 \end{proof}

\subsection{`Dressed' charged fields}

We now want to add radiation clouds to the Dirac fields. We first
treat the problem heuristically, and write the Dirac field in the
`integrational' notation as
 \mbox{$\Psi(f)=\int\ov{f(v)}\gamma\cdot v\,\Psi(v)\,d\mu(v)$}.
For each four-velocity of the particle $v$ we choose an
electromagnetic cloud profile $V_v(s,l)\in\V$, and form a modified
field
 \mbox{$\Psi(f,V_*)=\int\ov{f(v)}\gamma\cdot v\,W(V_v)\Psi(v)\,d\mu(v)$}.
This, of course, has only a heuristic value, but one can expect
that this field can be constructed in the von Neumann algebra
of a representation (from the class defining the
$C^*$-algebra~$\F$). Let us write, still at this informal
level, the commutation relation of this field with the
electromagnetic field. We find
\begin{equation}\label{heur}
 W(V_1)\Psi(f,V_*)=\Psi(S_{V_1,V_*}f,V_*)W(V_1)\,,
\end{equation}
 where
 $\big(S_{V_1,V_*}f\big)(v)=\exp\big[i\vp_{V_1,V_*}(v)\big]\,f(v)$
 with
\begin{equation}\label{Smod}
 \vp_{V_1,V_*}(v)
 =-\frac{e}{4\pi}\int\frac{v\cdot\Delta V_1(l)}{v\cdot l}\,d^2l
 +\{V_1,V_v\}\,.
\end{equation}

The problem of compensating the Coulomb field by the cloud field
in some region is now the problem of choosing $V_v$ so as to
compensate the first term in \eqref{Smod} by the second term,
for $V_1$ in some class. However, we note that the symplectic form
reduces to zero when restricted to any of the two subspaces of
functions $V(s,l)$ which are even or odd in $s$ respectively. But
$\Delta V_1(l)$ is the characteristic of the odd part of
$V_1(s,l)$. Thus the odd part of $V_v(s,l)$ has no influence on
this expected cancellation, and therefore may be assumed to
vanish. In consequence, $V_v$ has no long-range tail, and the
field $W(V_v)$ is infrared-regular. This brings in an important
simplification: in all representations in our class there is
$\pi(W(V_v))=\id_F\otimes\pi_r(W(V_v))$ and this operator is
independent of $\pi(\Psi(f))=\pi_F(\Psi(f))\otimes\id_r$. Our
informal modified field is now
 $\dsp\Psi_\pi(f,V_*)
 =\int\ov{f(v)}\gamma\cdot v\,\pi_F(\Psi(v))\otimes
 \pi_r(W(V_v))\,d\mu(v)$.

The use of representations for further construction is
unavoidable. We shall need some general additional assumptions
on their properties needed in the construction, as well as some
conditions on the `clouds' profiles $V_*$. We formulate these
assumptions in the present section successively, and test them
in a large class of representations in the next subsection.
\begin{assum}\label{measurable}
The profiles $V_v(s,l)\in\V$ are smooth functions of all their arguments $(v,s,l)$, even in $s$. For each pair of
vectors \mbox{$\vp,\chi\in\Hc_r$} the function
$v\mapsto(\vp,\pi_r(W(V_v))\chi)_r$ is measurable.
\end{assum}
\noindent Smoothness implies, in particular, that for each $V_1$ the function $\vp_{V_1,V_*}(v)$ in \eqref{Smod} is smooth, and the operator $S_{V_1,V_*}$ in \eqref{heur} is well defined in $\Sc(H_+)$.

Motivated by the above discussion we choose an orthonormal
basis $\{e_j\}$ of the Hilbert space $\K$ formed of functions $e_j\in\Sc(H_+)$, and `expand' $\pi_F(\Psi(v))$ in that basis. This
leads us to the definition
\begin{equation}\label{psipi}
 \Psi_\pi(f,V_*)=\sum_{j=1}^\infty\pi_F(\Psi(e_j))\otimes
 W_{\pi_r}(V_*,\ov{f}\,\Gamma e_j)\,,
\end{equation}
where $\Gamma$ is the operator defined by $(\Gamma
f)(v)=\gamma\cdot vf(v)$, and $W_{\pi_r}(V_*,\rho)$ is defined
by
\begin{equation}\label{Wweak}
 W_{\pi_r}(V_*,\rho)
 =\int\pi_r(W(V_v))\rho(v)\,d\mu(v)\,,
\end{equation}
integration in the weak sense: the operators are sandwiched in
$(\vp,.\,\chi)_r$ before integration. We note that
$|(\vp,\pi_r(W(V_v))\chi)_r|\leq\|\vp\|_r\|\chi\|_r$, so it is
sufficient that $\rho$ be integrable. Note also that all operators
$W_{\pi_r}(V_*,\rho)$ commute with each other, as all $V_v$ are
even.

\begin{prop}\label{psipiconv}
 The series defining $\Psi_\pi(f,V_*)$ by \eqref{psipi}
 converges ${}^*$-strongly to a~bounded operator independent of the
 choice of the basis $\{e_j\}$ in $\Sc(H_+)$.
\end{prop}
\begin{proof}
 If we denote by $\Psi^{(n)}_\pi(f,V_*)$ the series
 truncated to the first $n$ terms, set
 $C_{mn}=\Psi^{(n)}_\pi(f,V_*)-\Psi^{(m)}_\pi(f,V_*)$,
and use the anticommutation relations for $\Psi(e_j)$, we find
 \begin{equation}
 C_{mn}C_{mn}^*+C_{mn}^*C_{mn}=\id_F\,\otimes\!\sum_{j=m+1}^nw_j^*w^{}_j\,,
\end{equation}
where $w_j=W_{\pi_r}(V_*,\ov{f}\,\Gamma e_j)$. Now, using
\eqref{Wweak} it is easy to see that
\begin{multline}\label{Wj}
 \big(\vp,W_{\pi_r}(V_*,\ov{f}\,\Gamma e_j)\chi\big)_r
 =\Big(f,\big(\vp,\pi_r(W(V_*))\chi\big)_re_j\Big)\\
 =\Big(\ov{\big(\vp,\pi_r(W(V_*))\chi\big)_r}\,f,e_j\Big)\,,
\end{multline}
so if we choose any orthonormal basis $\vp_k$ of $\Hc_r$, we find
\begin{multline}
 \sum_{j=1}^\infty(\chi,w_j^*w_j\chi)_r=\sum_{j,k=1}^\infty|(\vp_k,w_j\chi)_r|^2\\
 =\sum_{k=1}^\infty\int|(\vp_k,\pi_r(W(V_v))\chi)_r|^2
 \ov{f(v)}\gamma\cdot v f(v)\,d\mu(v)=\|f\|^2\|\chi\|_r^2\,,
\end{multline}
the last step by the Lebesgue theorem. As $\sum_{j=1}^nw_j^*w_j$
is an increasing sequence of operators, this calculation shows
that $\sum_{j=1}^\infty w_j^*w_j=\|f\|^2\id_r$ in the
$\sigma$-strong sense. This is sufficient for the ${}^*$-strong convergence of the series \eqref{psipi} and the bound of the norm of the limit. The independence of the basis follows from the action of the limit operator on product vectors. It is easy to see with the use of \eqref{Wj} that
\begin{equation}\label{psipi2}
 (\xi_1\otimes\chi_1,\Psi_\pi(f,V_*)\,\xi_2\otimes\chi_2)
 =\Big(\xi_1,\pi_F\Big(\Psi\big(\,\ov{(\chi_1,\pi_r(W(V_*))\chi_2)_r}f\,\big)\Big)\xi_2\Big)_F\,.
\end{equation}
\end{proof}

The (anti-) commutation relations of the `dressed' Dirac fields
are:
\begin{gather}
 [\Psi_\pi(f,V_*),\Psi_\pi(f',V'_*)]_+=0\,,\label{psipsi}\\[1ex]
 [\Psi_\pi(f,V_*),\Psi_\pi(f',V'_*)^*]_+
 =\id_F\otimes W_{\pi_r}(V_*-V'_*,\ov{f'}\,\Gamma f)\,,\label{psipsistar}\\[1ex]
 \pi(W(V_1))\Psi_\pi(f,V_*)=\Psi_\pi(S_{V_1,V_*}f,V_*)\pi(W(V_1))\,,\label{Wpsi}
\end{gather}
where $S_{V_1,V_*}$ is given, as in the heuristic introduction,
by \eqref{Smod}. These relations are straightforwardly
calculated with the use of the definition \eqref{psipi}. For
the second and third identity use the technique of the above
proof and the independence of basis $\{e_j\}$ respectively.
Setting $V'_*=V_*$ we find that dressed fields with a fixed
profile $V_*$ satisfy the usual {\rm CAR} relations among
themselves. It follows thus by a~standard argument (see e.g.\
\cite{bra}) that $\|\Psi_\pi(f,V_*)\|=\|f\|$.

To investigate the long-range behaviour of the dressed fields, we
scale their radiation clouds. The profile $V_v$ in the element
$W(V_v)$ may be assumed to result from a conserved current $J_v$
supported in $R_{\gamma,\delta}$, having vanishing asymptote, even
with respect to the reflection: $J_v(-x)=J_v(x)$. As then, in
loose terms, $W(V_v)=\exp[-iA(J_v)]$ and $A(J_v)=\int
A(x)J_v(x)d^4x$, scaling the electromagnetic field observable to
spacelike infinity means replacing $J_v$ by
$J^R_v(x)=R^{-3}J_v(x/R)$ and taking the limit $R\to\infty$ (cf.\
\cite{bu86}). This scaling induces a simple scaling law for $V_v$.
Thus we set
\begin{equation}\label{VvR}
 V^R_v(s,l)=V_v(s/R,l)\,\\[1ex]
\end{equation}

\begin{assum}\label{wlimit}
There exist weak limits
\begin{equation}\label{ass2}
 \wlim_{R\to\infty}\pi_r(W(V^R_v))=\N_{\pi_r}(V_v)\,W_{\pi_r}^\infty(V_v)\,,
\end{equation}
such that $W_{\pi_r}^\infty(V_v)$ are unitary operators in
 $\Hc_r$, and the real, positive functions
 \mbox{$v\mapsto \N_{\pi_r}(V_v)>0$} are smooth and such that
 $1/\N_{\pi_r}(V_v)$ are multipliers\linebreak in~$\Sc(H_+)$.
\end{assum}
Note that it follows from Assumptions \ref{measurable} and
\ref{wlimit} that $\N_{\pi_r}(V_v)\leq1$ and functions
 $v\mapsto (\vp,W_{\pi_r}^\infty(V_v)\chi)$ are measurable for
 all $\vp,\chi\in\Hc_r$. Also, the operators
 $W_{\pi_r}^\infty(V_v)$ commute with each other.

 Mimicking the definitions \eqref{Wweak} and \eqref{psipi} we
 now define
\begin{gather}
 W_{\pi_r}^\infty(V_*,\rho)
 =\int W_{\pi_r}^\infty(V_v)\rho(v)\,d\mu(v)\,,\label{Wweakinf}\\
 \Psi_\pi^\infty(f,V_*)=\sum_{j=1}^\infty\pi_F(\Psi(e_j))\otimes
 W_{\pi_r}^\infty(V_*,\ov{f}\,\Gamma e_j)\,,\label{psipiinf}
\end{gather}
and note that also the analogue of \eqref{psipi2} holds:
\begin{equation}
 (\xi_1\otimes\chi_1,\Psi_\pi^\infty(f,V_*)\,\xi_2\otimes\chi_2)
 =\Big(\xi_1,\pi_F\Big(\Psi\big(\,\ov{(\chi_1,W_{\pi_r}^\infty(V_*) \chi_2)_r}f\,\big)\Big)\xi_2\Big)_F\,.
\end{equation}
The correctness and independence of basis of the definition
\eqref{psipiinf} is shown as in the proof of Proposition
\ref{psipiconv}. It is now easy to show that (the order of limits in the second relation is irrelevant)
\begin{gather}
 \wlim_{R\to\infty}W_{\pi_r}(V^R_*,\rho/\N_{\pi_r}(V_*))
 =W_{\pi_r}^\infty(V_*,\rho)\,,\label{Winf}\\[1ex]
 \begin{split}
 \wlim_{R\to\infty}\lim_{R'\to\infty}
 W_{\pi_r}\big(V^R_*-V'_*{}^{R'}&,\rho/[\N_{\pi_r}(V_*)\N_{\pi_r}(V'_*)]\big)\\[-1ex]
 &=\int W_{\pi_r}^\infty(V_v)W_{\pi_r}^\infty(V'_v)^*\,\rho(v)\,d\mu(v)\,,
 \end{split}\label{WinfWinf}
 \\[1ex]
 \wlim_{R\to\infty}\Psi_\pi(f/\N_{\pi_r}(V_*),V^R_*)
 =\Psi_\pi^\infty(f,V_*)\,;\label{psipiinf2}
\end{gather}
for the last relation use \eqref{psipi2} and the uniform
boundedness of the norms of the operators under the limit. To find
the (anti-) commutation relations of the dressed fields, we use
their representation \eqref{psipiinf2} and the relations
\eqref{psipsi} -- \eqref{Wpsi}, with the use of \eqref{WinfWinf}
on the r.h.s.\ of \eqref{psipsistar}. Setting now $V'_*=V_*$ we
find
\begin{gather}
 [\Psi_\pi^\infty(f,V_*),\Psi_\pi^\infty(f',V_*)]_+=0\,,\\[1ex]
 [\Psi_\pi^\infty(f,V_*),\Psi_\pi^\infty(f',V_*)^*]_+
 =(f,f')_\K \id\,,\\[1ex]
 \pi(W(V_1))\Psi_\pi^\infty(f,V_*)
 =\Psi_\pi^\infty(S^\infty_{V_1,V_*}f,V_*)\pi(W(V_1))\,,\label{WPsiinf}
\end{gather}
where
\begin{gather}
 \big(S^\infty_{V_1,V_*}f\big)(v)
 =\exp\big[i\vp^\infty_{V_1,V_*}(v)\big]\,f(v)\,,\label{SVinfty}\\[1.5ex]
 \vp^\infty_{V_1,V_*}(v)=-\frac{e}{4\pi}
 \int\frac{v\cdot\Delta V_1(l)}{v\cdot l}\,d^2l
 +\frac{1}{2\pi}\int V_v(0,l)\cdot\Delta V_1(l)\,d^2l\,.
\end{gather}
To show \eqref{WPsiinf}, one notes first that
$\dsp\lim_{R\to\infty}\vp_{V_1,V^R_*}(v)=\vp^\infty_{V_1,V_*}(v)$
and then observes that while taking the weak limit of
$\Psi_\pi(S_{V_1,V^R_*}f/\N_{\pi_r}(V_*),V^R_*)$ one can replace
$S_{V_1,V^R_*}$ by $S^\infty_{V_1,V_*}$ as the difference vanishes
in norm.

We note that the dependence of $\vp^\infty_{V_1,V_*}(v)$ on $V_1$
is only through its infrared tail $\Delta V_1$. In spacetime terms
it means that the dependence on the test current $J_1$ giving rise
to $V_1$ is only through its asymptote $J^\as_1$, which may be
assumed to be of the form $J^\as_1(z)=z\rho_1(z)$, in accordance
with Proposition~\ref{Jzg}. Thus using \eqref{DeltaV} we can write
\begin{gather}
 \vp^\infty_{V_1,V_*}(v)=\int\rho_1(z) F_v(z)\,d\nu(z)\,,\\
 F_v(z)=\frac{1}{2\pi}\int\frac{1}{z\cdot l}\,z\cdot\left[V_v(0,l)
 -\frac{e\,v}{2\,v\cdot l}\right]\,d^2l\,,\label{Fv}
\end{gather}
the second integral in the principal value sense.

The negative result mentioned at the beginning of Section
\ref{locdir} is now the following.
\begin{thm}
 There is no choice of profiles $V_v(0,l)$ such that
 $S^\infty_{V_1,V_*}f=f$ would hold for any test function $f$
 and for all $J^\as_1(z)=z\rho_1(z)$ supported in any given fixed
symmetrical spacelike cone.
\end{thm}
\begin{proof}
 The asymptote $\rho_1(z)$ is subject to two conditions: it must
be an even function and satisfy $\int\rho_1(z)d\nu(z)=0$ (cf.\
\eqref{contzero}). The only way to achieve
\mbox{$\exp[i\vp^\infty_{V_1,V_*}(v)]=1$} for some $v$ and all admissible $\rho_1$ supported in a given symmetrical spacelike cone would be that $F_v(z)=\con.$ on the patch of hyperboloid defining this cone (note that $F_v(z)$ is also even). This, however, is impossible for the following reason. It is
easily seen that $F_v(z)$ extends naturally to an even,
homogeneous function $F_v(x)$ of degree $0$ for all $x^2<0$ (by
simply replacing $z$ by $x$ in \eqref{Fv}). Now $F_v(z)=\con.$ in
a patch iff $F_v(x)=\con.$ in the corresponding cone. This,
however, is impossible, as we shall see that $\Box F_v(x)=2e/x^2$.
To show this, we first use the result of Appendix A of
\cite{her08}. Each possible profile $V_v(0,l)$ must be orthogonal
to $l$ and thus satisfies the conditions on $V(l)$ of this
Appendix. Using Eq.\,(A.4) one finds
\begin{multline}
 F^{(1)}_v(x)\equiv
 \frac{1}{2\pi}\int\frac{x\cdot V_v(0,l)}{x\cdot l}\,d^2l\\
 =-\frac{1}{2\pi}\int \p\cdot V_v(0,l)\,\log\frac{|x\cdot l|}{v\cdot l}\,d^2l
 +\frac{1}{2\pi}\int\frac{v\cdot V_v(0,l)}{v\cdot l}\,d^2l\,.
\end{multline}
This implies $\Box F^{(1)}_v(x)=0$. On the other hand one
explicitly calculates
\begin{equation}
 F^{(2)}_v(x)=-\frac{e\,v\cdot x}{4\pi}
 \int\frac{d^2l}{x\cdot l\,v\cdot l}
 =-e\, \frac{v\cdot x}{\sqrt{(v\cdot x)^2-x^2}}
 \artanh\frac{v\cdot x}{\sqrt{(v\cdot x)^2-x^2}}
\end{equation}
and $\Box F^{(2)}_v(x)=2e/x^2$.
\end{proof}
This result shows that it is impossible to choose $V_v(s,l)$ in
such a way that the exponential factor in \eqref{SVinfty} vanishes
for all test functions and test currents supported in symmetrical
spacelike cones. However, one can find $V_v(s,l)$ which makes this
exponential factor independent of $v$. Let
\begin{equation}
 V_v(s,l)=
 \left(\frac{e}{2}\right)\left(\frac{v}{v\cdot l}
 -\frac{t}{t\cdot l}\right)\eta\left(\frac{s}{t\cdot l}\right)\,,\label{profiles}
\end{equation}
where $t$ is a timelike unit vector and $\eta(s)$ is a smooth
function satisfying: \mbox{$0\leq \eta(s)\leq 1$}, $\eta(0)=1$,
$\eta(s)=\eta(-s)$ and there exist $s_0>0$ such that
$\eta(s)=0$ for $s>s_0$. For this profile, if it satisfies
Assumptions \ref{measurable} and \ref{wlimit} (beside smoothness, which is obvious), it follows that:
\begin{equation}
\vp^\infty_{V_1,V_*}=-\frac{e}{4\pi}
 \int\frac{t\cdot\Delta V_1(l)}{t\cdot l}\,d^2l\,,\qquad
 S^\infty_{V_1,V_*}=\exp\big[i\vp^\infty_{V_1,V_*}\big]\id\,.
\end{equation}
The commutation relation \eqref{WPsiinf} and its adjoint take
now the following simple form:
\begin{align}
 \pi(W(V_1))\Psi_\pi^\infty(f,V_*)
 &=e^{i\vp^\infty_{V_1,V_*}} \Psi_\pi^\infty(f,V_*)\pi(W(V_1))\,,\label{psipiV}\\
 \pi(W(V_1))\Psi_\pi^\infty(f,V_*)^*
 &=e^{-i\vp^\infty_{V_1,V_*}} \Psi_\pi^\infty(f,V_*)^*\pi(W(V_1))\,.\label{psipiVstar}
\end{align}
It is now possible to restrict the scope of test functions $f$
to those resulting from compactly supported four-spinor test
fields $\chi$ in \eqref{fchi}. Then the observables
$\Psi_\pi^\infty(f,V_*)^*\Psi_\pi^\infty(f',V_*)$ form a local
net commuting with the electromagnetic field, with localization
determined by the union of the supports of $\chi$ and $\chi'$.
These are the asymptotic incarnations, in our model, of the
quantities discussed at the beginning of Section \ref{locdir}.

\begin{assum}\label{VV'}
For any two profiles $V_v$, $V_v'$ of the form \eqref{profiles}
(with possibly different vectors $t$ and functions $\eta$) the
unitary operator
$W^\infty_{\pi_r}(V_v)W^\infty_{\pi_r}(V'_v)^*$ formed by the
operators defined by Assumption \ref{wlimit} is independent of
$v$.
\end{assum}

With this assumption it is now easy to see that the observables
defined above do not depend on a particular choice of the
profile $V_v$ in the class \eqref{profiles}.

\subsection{Special choice of representation}\label{repr}

In this subsection we show that Assumptions \ref{measurable},
\ref{wlimit} and \ref{VV'} are fulfilled for profiles
\eqref{profiles} in a class of representations $\pi_r$ in
\eqref{rep} constructed in earlier papers.

Consider the vector space of equivalence classes of real,
smooth vector functions $f_a(l)$ on the cone, homogeneous of
degree $-1$, with $l\cdot f(l)=0$. The equivalence relation is
introduced by: $f_{1}\sim f_2 \Leftrightarrow
f_{1a}(l)=f_{2a}(l)+\beta(l)l_a$. The completion of this space
with respect to the scalar product
\[
 (f_1,f_2)_0=-\int f_1(l)\cdot f_2(l)\,d^2l
\]
is a~real Hilbert space denoted $\Hc_0$. The closure of the
subspace of (equivalence classes of) smooth functions satisfying
$L\wedge f=0$ forms a Hilbert space denoted by $\Hc_{IR}$. Let
$H(s,l)$ be a homogeneous of degree 0, smooth function, such that
$\lim\limits_{s\rightarrow \pm\infty}H(s,l)=\pm1$ and
$\dot{H}(s,l)$ satisfies the falloff condition analogous to
\eqref{falloff1}. We denote $h(s,l)=\pi \dot{H}(s,l)$ and fix
notation for Fourier transform with respect to $s$ by
\begin{equation}
 \ti{h}(\w,l)=\frac{1}{2\pi}\int e^{i\w s}h(s,l)\,ds\,,\label{fourier}
\end{equation}
so $\ti{h}(0,l)=1$. Following the notation of \cite{her98} and
\cite{her08} we set
\[
 p(\dV)=\ti{\dV}(0,l)=\frac{1}{2\pi}\Delta V\,,
\]
the long range characteristic of $V(s,l)$, and denote by
$r_h(\dV)$ the orthogonal projection of
 $\frac{1}{2}\int \dV H(s,l)ds$ onto~$\Hc_{IR}$.
We split function $V(s,l)$ into the IR-regular and IR-singular
part by setting:
\begin{multline}
 \ti{\dV}(\w,l)
 =\Big[\ti{\dV}(\w,l)-\ti{\dV}(0,l)\ti{h}(\w,l)\Big]
 +\ti{\dV}(0,l)\ti{h}(\w,l)\\
 =\ti{\dV}_\reg(\w,l)+\ti{\dV}(0,l)\ti{h}(\w,l)\,.\label{divideV}
\end{multline}
In particular, $p(\dV)=0$ means that $\dV$ is IR-regular, i.e.
the field has no `long range tail'.

Further, we consider the Weyl algebra generated by the
elements\linebreak \mbox{$w(g\oplus k)$}, where $g\oplus k$
belongs to the vector space $\C^{\infty}_{IR}\oplus\C^\infty_{IR}$
($\C^\infty_{IR}:=\C^\infty\cap\Hc_{IR}$, differentiability is
understood in the sense of $L_{ab}$) with the symplectic
structure:
\begin{equation}
 \{g_1\oplus k_1,g_2\oplus k_2\}_{IR}
 := (g_1,k_2)_{IR}-(k_1,g_2)_{IR}\,.
\end{equation}
Algebraic relations satisfied by elements $w(g\oplus k)$ are
\begin{align}
 w(g_1\oplus k_1)w(g_2\oplus k_2)
 &=e^{-\frac{i}{2}\{g_1\oplus k_1,g_2\oplus k_2\}_{IR}}
 w\big((g_1+g_2)\oplus(k_1+k_2)\big)\label{w1}\,,\\
 w(g\oplus k)^*&=w(-(g\oplus k))\,.\label{w2}
\end{align}

Let $\pi_{\mathrm{sing}}$ be a cyclic representation of this
algebra derived by GNS construction from the state
\begin{gather}
 \w_\mathrm{sing}\big(\w(g\oplus k)\big)=
 \exp\Big(-\tfrac{1}{4}(g,C^{-1}g)_{IR}-\tfrac{1}{4}(k,Ck)_{IR}\Big)\,,\label{wsing}
\end{gather}
with the corresponding Hilbert space $\Hc_\sing$ and the cyclic vector
$\W_\sing$.\linebreak Here $C$ is any positive, trace-class
operator such that $\C^\infty_{IR}\subset C^{1/2}\Hc_{IR}$,
$\ov{C^{-1/2}\C^\infty_{IR}}^{\Hc_{IR}}=\Hc_{IR}$. Denote by
$\pi_0$ the standard positive energy Fock representation of
infrared-regular fields, generated by GNS construction from the
vacuum state
\begin{gather}
 \w_0(W(V_\reg))
 =\exp\left(-\tfrac{1}{2}F(\dV_\reg,\dV_\reg)\right)\,,\\[1ex]
 F(\dV_1,\dV_2)=\int_{\w\geq0}\left(-\ov{\ti{\dV}_1(\w,l)}\cdot
 \ti{\dV}_2(\w,l)\right)\frac{d\w}{\w}d^2l\,, \label{FVV}
\end{gather}
with the corresponding Hilbert space $\Hc_\reg$ and cyclic
vector $\W_0$. Then the formula
\begin{equation}
 \pi_r(W(V))=\pi_\sing\big(w(p(\dV)\oplus
 r_h(\dV))\big)\otimes\pi_0(W(V_\reg)),\label{pir}
\end{equation}
determines  a regular, translationally covariant positive energy
representation of $\B^-$ on
$\Hc_r=\Hc_{\mathrm{sing}}\otimes\Hc_\reg$ \cite{her98}. Now one
has to prove that assumptions \ref{measurable}, \ref{wlimit} and
\ref{VV'} are fulfilled for this choice of $\pi_r$.

It was shown in \cite{her98} that the representation $\pi_r$
does not depend on the concrete shape of $H(s,l)$. Therefore,
for the convenience of the proof of proposition \ref{exrep}, we
shall assume, from now on, a special choice of this function.
We put $H(s,l)=H_t\left(s/t\cdot l\right)$ for a
timelike unit vector $t$, and a smooth function $H_t$ such that
for some $u_0>0$ there is $H_t(u)=1$ for $u>u_0$, and
$H_t(u)=-1$ for $u<-u_0$.
\begin{prop}\label{exrep}
For the representation $\pi_r$ defined by \eqref{pir} and the
profiles $V_v$ given by \eqref{profiles}, Assumptions
\ref{measurable}, \ref{wlimit} and \ref{VV'} are satisfied.
\end{prop}
\begin{proof}
(Assumption \ref{measurable})  To prove the measurability, it
suffices to show that $(y,\pi_r(W(V_v))x)_r$ is continuous in
$v$ for vectors from a total set, those of the form
\begin{align*}
 &x=\pi_\sing\big(w(g_1\oplus k_1)\big)\W_{\sing}\otimes \pi_0\big(W(V_1)\big)\W_0\,,\\
 &y=\pi_\sing\big(w(g_2\oplus  k_2)\big)\W_{\sing}\otimes \pi_0\big(W(V_2)\big)\W_0\,,
\end{align*}
where $V_i$, $i=1,2$, are IR-regular. As~$V_v$~is
IR-regular, so
\begin{equation}
\pi_r(W(V_v))=\pi_\sing(w(0\oplus r_h(\dV_v)))\otimes\pi_0(W(V_v))\,.\label{piWVv}
\end{equation}
One obtains:
\begin{multline}
 (y,\pi_r(W(V_v))x)_r
 =\w_\sing\big(w(g_2\oplus k_2)^*w(0\oplus r_h(\dV_v))w(g_1\oplus k_1)\big)\times\\
 \times \w_0\big(W(-V_2)W(V_{v})W(V_1)\big)\,.\label{yWVvx}
\end{multline}
From the algebraic relations it follows that:
\begin{multline}
 \w_0\big(W(-V_2)W(V_v)W(V_1)\big)=\\
 =\exp\Big[-\tfrac{1}{2}F(\dV_1-\dV_2+\dV_v,\dV_1-\dV_2+\dV_v)
 -\tfrac{i}{2}\{V_v,V_1+V_2\}-\tfrac{i}{2}\{V_1,V_2\}\Big]\,.\label{comm}
\end{multline}
Since $F(\dV_v,\dV_v)$ and $F(\dV_v,\dV_k)$ ($k=1,2$), as easily shown, are smooth in $v$, so is the r.h.s. of \eqref{comm}. Now we turn to $\w_\sing$.  Using \eqref{w1},
\eqref{w2} and \eqref{wsing}, one finds:
\begin{multline}
 \w_\sing\big(w(g_2\oplus k_2)^*w(0\oplus r_h(\dV_v))w(g_1\oplus k_1)\big)=\\
 \exp\Big[-\tfrac{1}{4}\big(\Delta g,C^{-1}\Delta g\big)_{IR}-\tfrac{1}{4}\big(\Delta k+r_h(\dV_v),
 C[\Delta k+r_h(\dV_v)]\big)_{IR}\Big]\\
 \times \exp\Big[\tfrac{i}{2}\big(r_h(\dV_v),g_1+g_2\big)_{IR}
 +\tfrac{i}{2}(g_2,k_1)_{IR}-\tfrac{i}{2}(g_1,k_2)_{IR}\Big]\,,\label{wsingx}
\end{multline}
where $\Delta g=g_1-g_2$, $\Delta k=k_1-k_2$.
To prove that the r.h.s of \eqref{wsingx} is indeed a continuous
function  in $v$, it suffices to show that terms of the form:
$\big(r_h(\dV_v),C\,r_h(\dV_v)\big)_{IR}$, $\big(r_h(\dV_v),k\big)_{IR}$,
$\big(k,C\,r_h(\dV_v)\big)_{IR}$ are continuous in $v$ for
$k\in\C_{IR}^{\infty}$. As $C$ is a bounded operator, it is
sufficient to show that $r_h(\dV_v)$, as an
element of $\Hc_{IR}$, is norm-continuous in $v$ . Since
\begin{equation}\label{rhVv}
 \tfrac{1}{2}\int \dV_v(s,l)H_t\left(\frac{s}{t\cdot l}\right)ds
 =\frac{e}{4}\int\dot{\eta}(u)H_t(u)du\,
 \Big(\frac{v}{v\cdot l}-\frac{t}{t\cdot l}\Big)=r_h(\dV_v)(l)\,,
\end{equation}
we have:
\begin{equation}
 ||r_h(\dV_v)-r_h(\dV_{v'})||_{IR}^2
 =\left(\frac{e}{4}\int \dot{\eta}(u)H_t(u)du\right)^2
 \bigg[-\int\left(\frac{v}{v\cdot l}-\frac{v'}{v'\cdot l}\right)^2 d^2l\,\bigg]\,.
\end{equation}
The last integral can be calculated explicitly:
\begin{multline}
 -\int\Big(\frac{v}{v\cdot l}-\frac{v'}{v'\cdot l}\Big)^2d^2l
 =\int\Big[2\frac{v\cdot v'}{(v\cdot l)(v'\cdot l)}
 -\frac{1}{(v\cdot l)^2}-\frac{1}{(v'\cdot
 l)^2}\Big]\,d^2l=\\[.5ex]
 =8\pi\bigg\{\frac{v\cdot v'}{\sqrt{(v\cdot v')^2-1}}
 \log\left(v\cdot v'+\sqrt{(v\cdot
 v')^2-1}\right)-1\bigg\}\label{intd2l}\,.
\end{multline}
Because \eqref{intd2l} converges to $0$ for $v\rightarrow v'$,
$r_h(\dV_v)$ is norm  continuous. Finally we can conclude that
\eqref{wsingx} is a continuous function of $v$. This ends the
proof of Assumption \ref{measurable}.
\\
\hspace*{1em}(Assumptions \ref{wlimit} and \ref{VV'}) First we
show the existence of the weak limit
\mbox{$\wlim_{R\to\infty}\pi_r(W(V^R_v))$}. The norms of
$\pi_r(W(V^R_v))$ are uniformly bounded, so it is sufficient to
obtain the weak limit for operators sandwiched between vectors
from a total set chosen as in the proof of Assumption
\ref{measurable}. We have to investigate the limit of
the expressions \eqref{comm} and \eqref{wsingx} in which $V_v$ has been replaced by $V^R_v$, for  $R\rightarrow\infty$.
From \eqref{VvR} and \eqref{fourier} one has $\ti{\dV}{}_v^R(\w,l)=\ti{\dV}_v(R\w,l)$.
As $\ti{\dV}_k(0,l)=0$, $k=1,2$, it follows by the Lebesgue dominated convergence theorem that $\lim\limits_{R\rightarrow\infty}F(\dV^{R}_{v},\dV_k)=0$ (see \eqref{FVV}), and
since $\{V^{R}_{v},V_k\}=2\,\mathfrak{Im}
\big(F(\dV^{R}_{v},\dV_k)\big)$, also
$\lim\limits_{R\rightarrow\infty}\{V^{R}_{v},V_k\}=0$.
On the other hand, by a change of the integration variable $\w$ one finds
\begin{equation}
F(\dV^{R}_{v},\dV^{R}_{v})
=F(\dV_{v},\dV_{v})\,.
\end{equation}
In this way, for the scaled version of  \eqref{comm} we obtain:
\begin{equation}
 \lim\limits_{R\rightarrow\infty}\w_0\left(W(-V_2)W(V^R_{p})W(V_1)\right)
 =\N_{\pi_r}(V_v)\,\w_0(W(-V_2)W(V_1))\,,
\end{equation}
where
\begin{equation}\label{Nmod}
 \N_{\pi_r}(V_v)=\exp\Big(\tfrac{1}{2}
 \int_{\w\geq0}\ov{\ti{\dV}_v(\w,l)}\cdot\ti{\dV}_{v}(\w,l)\frac{d\w}{\w}d^2l\Big)\,.
\end{equation}
Thus
\begin{equation}
 \wlim\limits_{R\to\infty}\pi_0(W(V^R_v))=\N_{\pi_r}(V_v)\,\id\,.
\end{equation}
For the IR-singular part we note that
\begin{equation}
 \lim_{R\to\infty}\|r_h(\dV_v^R)+V_v(0,.)\|_{IR}=0\,,
\end{equation}
which is easily shown with the use of \eqref{rhVv}. Thus using the scaled version of \eqref{wsingx} we find
\begin{equation}
 \wlim\limits_{R\to\infty}\pi_\sing(w(0\oplus r_h(\dV^R_v)))
 =\pi_\sing(w(0\oplus-V_v(0,.)))\,.
\end{equation}
Therefore, we can finally conclude that the relation
\eqref{ass2} is satisfied, with $\N_{\pi_r}$ given by
\eqref{Nmod}, and
\begin{equation*}
 W_{\pi_r}^\infty(V_v)=\pi_\sing(w(0\oplus-V_{v}(0,.\,))\otimes\id\,.
\end{equation*}
This form of these operators ensures that Assumption \ref{VV'} is satisfied. After a suitable change of  variables one finds that the factor function has the form
\begin{equation}
 \N_{\pi_r}(V_v)\,
 =\exp\left(\frac{e^2}{8}\int_{u\geq0}u|\ti{\eta}(u)|^2du\,
 \int\Big(\frac{v}{v\cdot l}-\frac{t}{t\cdot l}\Big)^2d^2l\right)\,,
\end{equation}
where $\ti{\eta}$ is the Fourier transform of $\eta$ defined as
in \eqref{fourier}. Using \eqref{intd2l} we obtain:
\begin{equation}
 \N_{\pi_r}(V_v)\,
 =\exp\bigg\{-c\bigg[\frac{v\cdot t}{\sqrt{(v\cdot t)^2-1}}
 \log\left(v\cdot t+\sqrt{(v\cdot t)^2-1}\right)-1\bigg]\bigg\}\,,
\end{equation}
where $c> 0$ is a constant. The function \mbox{$v\mapsto
\N_{\pi_r}(V_v)$} is smooth and for\linebreak $v^0\rightarrow
\infty$ we have: \mbox{$1/\N_{\pi_r}(V_v)\sim \con(v^0)^c$}, with similar estimates for derivatives.
This proves that $1/\N_{\pi_r}(V_v)$ are multipliers in
$\Sc(H_+)$.
\end{proof}

\setcounter{equation}{0}

\section{Conclusions}

The algebra proposed earlier for the description of asymptotic
fields in spinor electrodynamics incorporates Gauss' law and thus
has good chances to form (at least a substantial part of)
a~consistent model of the long-range
behavior of QED. We have found here how to give the elements of
this field algebra localization in regions contained in an
arbitrarily chosen time slice `fattening towards edges'. Compact
localization regions may be chosen only for infrared-regular
electromagnetic fields. Both infrared-singular electromagnetic
fields as well as charged fields have always localization regions
extending to spacelike infinity. However, the infrared singular
electromagnetic fields may be decomposed into fields localized in
arbitrarily `thin' fattened symmetrical spacelike cones. On the
other hand we have found that there is no way of attaching an
infrared cloud to the charged field so as to localize it in such
region, at least in a wide class of representations which satisfy
some natural general conditions. Nevertheless, we have also shown
that compactly supported observables may be formed by simple
multiplication of appropriately dressed charged fields with
compensating charges.

The lack of spacelike-cone localization of dressed Dirac fields in the present model seems to be nonstandard, as already mentioned in Introduction and Section~5. One could object that the model, although it incorporates global Gauss' law, still lacks some additional asymptotic electromagnetic variables. The construction of the model suggests that in such case the variables would have to originate as limits of gauge-dependent local electromagnetic potentials. However, whether the model is indeed incomplete can only be decided by finding its place in a formulation of fully interacting electrodynamics. In particular, it would be interesting to formulate a perturbative electrodynamics incorporating some nonperturbative infrared aspects of the present model.

On the other hand, we would like to stress once more a physically important aspect of the model considered here. Our fundamental fermion fields are genuinely charged, satisfying Gauss' law even before `dressing'. Dressing is considered for the sake of inducing a
certain localization of these fields, as well as an auxiliary step
in the construction of bi-fermion observables. Simplified as the
model is, it is at the same time non-perturbative.

This is to be contrasted with all forms of `dressing' of fermion
fields in local formulations of QED. There, in the indefinite
metric space (Gupta-Bleuler), local Dirac fields cannot carry
physical charge, as they commute with the electric flux at spatial
infinity. After constructing a perturbative solution of an initial
theory formulated in such space, one attempts then, by the
addition of Lorenz condition and nonlocal dressing of charged
fields, to restore Maxwell equations and transport the theory into
a Hilbert space of physical vector states. The dressing takes the
form of a formal local gauge transformation in which the gauge
function is constructed with the use of electromagnetic potential
(see e.g. \cite{sym}). In an Ansatz put forward by Dirac this has
the following form:
\begin{equation}\label{dress}
 \Psi(x)=\exp[ieG(x)]\psi(x)\,,\quad
 \mathcal{A}(x)=A(x)-\partial G(x)\,,
\end{equation}
where $G(x)=\int r^a(x-y)A_a(y)d^4y$; here $r^a(x)$ is a~vector
distribution satisfying $\partial_ar^a(x)=\delta^4(x)$. Within
perturbative approach to QED this idea has been implemented most
rigorously in the `axiomatic perturbative' formulation by
Steinmann~\cite{ste}. In this approach the above tentative
transformation is carried out not on the level of fields, but
rather Wightman functions. As argued by Steinmann, the results are
insensitive to a choice of a particular form of the distribution
$r^a$. And as among such distributions are some with supports in
spacelike cones, one can argue that in this way charged fields may
be pushed into such regions.

These constructions, rigorous as they are within the limits of the
procedure followed in this approach, are not without weak points.
First, not only the local interaction, but also the dressing
exponent is treated perturbatively; this is admitted by Steinmann
himself to be an obstacle to a completely reliable representation
of the infrared problems. Secondly, the dressing transformation
\eqref{dress} is infrared-singular and cannot be performed in this
form even at the level of Wightmann functions; the actual way it
is done, is via an effective spatial truncation followed by an
adiabatic limit. However, precisely these two points are of
critical importance for the infrared problem.

Finally, we want to comment on our choice of representations. One cannot exclude that the use of some more infrared singular representations would modify our results. That localization may be improved `in front of' infravacua (KPR-type representations \cite{KPR}) has been shown by Kuhnhardt \cite{Kun} in a scalar model due to Buchholz \emph{et al.} \cite{BDMRS}. One of the main motivations for the introduction of such more singular representations of free electromagnetic fields is the fact that they may be stable under the addition of radiation fields produced in scattering processes. However, in this connection we want to mention two facts on the asymptotic model considered here. First, it has been shown in \cite{her08} that representations discussed above in Section \ref{repr} do suffice to absorb radiation fields produced by a classical current. Second, in this model the \mbox{asymptotic} fields are not completely decoupled, and the electric flux at spatial infinity is due both to free as well as Coulomb parts. However, the electric flux of the total field at infinity is an invariant characteristic of the process, not changing with time (the asymptotic flux depends on the spacelike direction, but, in fact, is invariant under any finite spacetime translation of the point from which we go to spacelike infinity). This is a fact in classical theory, and should be also expected in the full quantum theory.

\section*{Appendix}

\appendix

We prove here the relation \eqref{theta}. Let $y\in R_\delta$
and $|x-y|\leq\ga$ and denote
\mbox{$\ka=\sqrt{(1+\delta^2)/(1-\delta^2)}$}, $-y^2=r^2$,
$r>0$. Then
\begin{equation*}
 |x^2-y^2|\leq|(x-y)^2|+2|y\cdot(x-y)|\leq\ga^2+2\ga\ka r\,,
\end{equation*}
where we have used \eqref{loreu}. Thus
 $-r^2-2\ga\ka r-\ga^2\leq x^2\leq-r^2+2\ga\ka r+\ga^2$.
Consider now two cases.

 (i) $y^2+R^2\geq0$ and $x^2+R^2\leq0$.\\
It follows that $R^2-r^2\geq0$ and
 $-r^2-2\ga\ka r-\ga^2+R^2\leq0$, so $r\in\<R-R_1,R\>$, with
 $R_1=\ga\ka$ (although not the whole interval is covered).

 (ii) $y^2+R^2\leq0$ and $x^2+R^2\geq0$.\\
 It follows that $R^2-r^2\leq0$ and
  $-r^2+2\ga\ka r+\ga^2+R^2\geq0$, so $r\in\<R,R+R_2\>$, with
 $R_2=\ga(\ka+\sqrt{\ka^2+1})$ (with the same remark as above).

Summarizing, we have that from $(x^2+R^2)(y^2+R^2)\leq0$ it
follows \mbox{$-(R+R_2)^2\leq y^2\leq-(R-R_1)^2$} for $R\geq
R_1$, which implies \eqref{theta}.

\end{document}